\documentclass{article}
\usepackage{geometry}
\usepackage{}
\usepackage{amssymb}
\usepackage{mathrsfs}
\usepackage{graphicx}
\usepackage{subfigure}
\usepackage{amsmath}

\usepackage{tabu}
\usepackage[linesnumbered,ruled]{algorithm2e}
\usepackage{multirow}
\usepackage{cases}
\usepackage[noblocks]{authblk}

\newtheorem{lemma}{{Lemma}}
\newtheorem{theorem}{Theorem}
\newtheorem{corollary}{Corollary}

\newtheorem{definition}{Definition}

\begin{document}

\title{Synthesis of Quantum Images Using Phase Rotation}

\author[1]{Shiping Du}
\author[1,]{Daowen Qiu \footnote{ Daowen Qiu is corresponding author.}}
\author[2]{Jozef Gruska}
\author[3]{Paulo Mateus}
\affil[1]{Department of Computer Science, Sun Yat-sen University, Guangzhou 510006, China.}
\affil[2]{Faculty of Informatics, Masaryk University, Brno, Czech Republic.}
\affil[3]{SQIG--Instituto de Telecomunica\c{c}\~{o}es, Departamento de Matem\'{a}tica, Instituto Superior T\'{e}cnico, Av. Rovisco Pais 1049-001, Lisbon, Portugal.
}

\date{}

\maketitle

\begin{abstract}
 A topic about synthesis of quantum images is proposed, and a specific  phase rotation transform constructed is adopted to theoretically realise the synthesis of two quantum images.
 The synthesis strategy of quantum images comprises three steps, which include:
 (1) In the stage of phase extraction, we obtain the phases of the state of the quantum image by transforming the state of the quantum image to prepare the conditions for multiple phases extraction.
 (2) In the stage of rotation operator construction, the phases obtained in the first stage are used to construct the rotation operator where a mechanism is introduced into it to reduce the phase overflow.
 (3) In the stage of application of the rotation operator, we apply the operator constructed in the second stage on the state of quantum image to get a goal state.
 Additionally, numerical analysis gives the joint uncertainty relation of the pixel of the synthesized quantum image.
 The analysis result about the compression ratio indicates that the phase rotation transform and the overflow control mechanism are effective.
%

$\mathbf{Keywords}$ {Quantum image, image synthesis, phase errors, multiple phase estimation}
\end{abstract}

\section{Introduction}
\label{introduction}
Since quantum world tools have been shown to be more powerful than tools of the classical world in many areas, it is natural, and actually very important, to try to explore \cite{a8,a28,a29}
in depth the methods and properties of the quantum image processing.

Quantum images have already been prepared with several mature technologies using different methods in laboratories. For example,  four wave mixing (FWM) \cite{a15} is the widely used one \cite{a16}.
On the other side, in spite of the fact that  quantum images can be easily prepared in physical laboratory, but how to specify the information of a quantum images in a quantum computer \cite{a10,a11} remains a problem.
Recently, some papers have discussed relevant topics, such as \cite{a1,a2,c1,c2}, where the authors propose or summarize different mathematic forms of the state representation of quantum images.
Typically, for example, flexible representation of quantum images (FRQI) \cite{a1,a2} which uses a single qubit to encode the grey level, novel enhanced quantum representation(NEQR) \cite{a17} which improves the expression with two qubits,binary key image generation \cite{b3}, and flexible quantum representation for grey-level quantum images (FQRGI) \cite{a6} et al.
Based on these types of state representation of quantum image,  different operations to the image are explored, such as \cite{c3,b1,b2,b4,b5,b6,c4}. Jiang et al. \cite{c4} creatively propose a  quantum version algorithm based on the improved NEQR \cite{a17} to implement the scaling of the quantum image. Caraiman et al. \cite{c7} realise the image segmentation on a quantum computer.

Image synthesis is an important topic in  area of image processing. One of the most important applications of image synthesis is information hiding \cite{b7,b8,b9,b10}.
Song et al. \cite{a3} proposed an algorithm for information hiding via introducing the quantum image embedding operation while the images are represented as a watermark and a carrier.
Actually, since classical synthesis has widely applications in the reality \cite{c5,c6}, we believe that more quantum image synthesis algorithms to be explored in the future have great significance.

In the present paper, a concrete procedure of image synthesis is described.
Since the image is constituted with pixels, the essence of quantum image synthesis is actually  pixel synthesis.
That's why our paper could analyzes the joint uncertainty of the synthesized pixel in section \ref{joint_uncertainty}.

One basic problem for quantum image synthesis is handling of the brightness of pixels.
Based on the state representation of quantum images and phase rotation operation adopted, the phases are the only available parameters to be chosen for implementing the grey-level control, so the image synthesis is equivalent to the corresponding pixel synthesis of two quantum images, and  pixels synthesis is equivalent to the synthesis of two phases.
 So how to ensure the precision of the phase extracting from the state of the quantum image, and how to effectively restrain the phases  to make them in a specified range are the critical point for us to get a synthesized image successfully. The analysis about effectiveness of overflow control is in section \ref{restriain}.


The rest of the paper is organized as follows.
Section \ref{method} introduces the preliminary knowledge about the state representation of the quantum images.
Section \ref{embed} describes the algorithm of quantum image synthesis.
Section \ref{analysis} analyzes the joint uncertainty of the pixel of the synthesized image, and discuss the effectiveness of phase compression. The last two parts are comparison and conclusion, respectively.

\section{Preliminary}
\label{method}

\subsection{Fundamental theory of quantum image transform}
\label{fundamental_theo}
The state representation of quantum images has to be consistent with images produced by technologies and convenient for the theoretical computation in the future. Concretely, the state representation of quantum images should captures and reflects the following facts.

$\mathbf{Facts}$

 (a) Quantum image can be expressed, using complex amplitudes mathematical expression, in quantum optics.

 (b) Quantum image is composed of pixels \cite{a5}.

 (c) The brightness of each pixel point is associated with the fluctuation of intensity of the photons \cite{a5}.

 (d) Each pixel  is associated with a coordinate.

In the light of the authoritative view of Kolobov \cite{a5},
the brightness of each pixel of a quantum image is determined by the intensity of photons.  Meanwhile, pixel defined by Kolobov \cite{a5} is the integration of photon number within a certain time period at this pixel point. However, to operate on all photons in one pixel point individually and simultaneously is difficult in reality.

Considering to store the pixel information in the phases, and controlling the brightness of each pixel point through modulating the phases of the state of quantum image, we can realise effectively control to the pixels of the quantum image.
Assume that the phases $\theta_i \in (0,\pi)$ are from the phase space $\mathbf{\theta}$, and let $\eta_i $ be from the photon number space $\eta$ representing the number of photons ( $i$ labels the different colors of the quantum image, so the maximal $i$ is the color type defined in the quantum image).
Defining a mapping $\mathbf{f}$ from the variables $\theta_i$ to $\eta_i$
\begin{eqnarray} \label{relation_a}
\mathbf{f}:\theta_i \rightarrow  \eta_i.
 \end{eqnarray}
Obviously, corresponding to $\theta_i$, $\eta_i$ also represents the brightness of the pixel point.
For $\theta_i$, motivated by the $\mathbf{Facts}$ (b) and (c) and Eq. \ref{relation_a},
as long as we know the phase, the number of  photons in one pixel point is also determined.
In addition, it is necessary to assume that the function $\mathbf{f}$  is monotonically increasing.

 Now,  we introduce the following state $|I(\theta)\rangle$ \cite{c8} as the state representation of a quantum image:
\begin{eqnarray} \label{eq_def}
|I(\{\theta_j\})\rangle = \frac{1}{2^{n}}  \sum_{j=0}^{2^{2n}-1}(|0\rangle + e^{i\theta_j}|1\rangle) \bigotimes |j\rangle,
 \end{eqnarray}
where $\{\theta_j\} =  (\theta_0, \theta_2, ..., \theta_{2^{2n}-1}) $, $\theta_j \in (0,\pi/2), j=0,1,..., 2^{2n}-1$. The relative phase information $\theta_j$ in $|0\rangle + e^{i\theta_j}|1\rangle$ encodes the grey level.
$|j\rangle$, $j \in \{ 0,1,..., 2^{2n}-1 \}$, is a $2^{2n}$ dimensional basis state,
and $|j\rangle$ represents the coordinate of $j^{th}$ pixel point in the pixel matrix of a quantum image.
The feature of state of quantum images in  Eq. (\ref{eq_def})  indicates that, the synthesis of  quantum image can be implemented only through  phase rotation. Therefore, to extract  phases from the state of  quantum images is the first and necessary step.

\section{Synthesis of quantum images}
\label{embed}

Given two quantum images,  say  a carrier and embedder (a image to be embedded). The synthesis will lead to the pixel accumulation of the carrier  and the embedder, and give rise to  the brightness changed in the overlapped positions correspondingly.

Since the pixels of the quantum image are represented with phases, so we should obtain all phases of the quantum image before starting to synthesize two quantum images.
In our paper, MPE (for details, see also \cite{a23,a24}) is considered to be an effective method to obtain phases for the given state of quantum images.
Fig. \ref{scheme} shows the procedure of synthesizing two images. There are three steps needed in total (Appendix \ref{shortcoming} lists different cases of phase overflow. This is why we correct the rotation operator in algorithm \ref{algorithm_1}).
The synthesis algorithm of quantum images can be described with algorithm \ref{algorithm_1}, where $|I(\{ \theta_j \})\rangle$ and $|I(\{\varphi_j \})\rangle$ are the states of quantum images defined as Eq. (\ref{eq_def}).

\begin{figure} 
\centering
{\includegraphics[height=3cm,width=7cm]{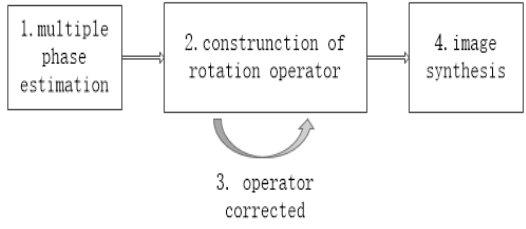}}
\caption{\label{scheme} Scheme of synthesis of two quantum images using phase rotation. There are three steps in the algorithm. That is, multiple phase extraction, construction of rotation operator, and image synthesis. Especially, rotation operator construction should be corrected by embedding  the overflow control mechanism.}
\end{figure}

\begin{algorithm}[ht]\label{algorithm_1}
\caption{Synthesis($|I(\{\theta_j \})\rangle, |I( \{ \varphi_j \})\rangle$)}
\KwIn{ Two states of quantum images to be synthesized: $|I( \{ \theta_j \})\rangle, |I( \{ \varphi_j \})\rangle$ .}  \KwOut{The state of synthesized quantum image .}
Extracting phases from two states $|I( \{ \theta_j \})\rangle, |I( \{ \varphi_j \})\rangle$ with MPE\;
Rotation operator $U'$ constructed and corrected\;
Applying $U'$ on the state $|I( \{ \varphi_j \})\rangle$ \;
\end{algorithm}

We have the following instructions:

Step 1 completes phases extraction. Actually, we should cope with many same quantum states $|I( \{ \theta_j \})\rangle$ and $|I( \{ \varphi_j \})\rangle$, then we could obtain these phases from two different states.

Step 2 constructs rotation operator which is embedded the internal error correction mechanism.

Step 3 completes the synthesis of the state of quantum images. The state $|I( \{ \varphi_j \})\rangle$ used in this step is also same as the backup in the first step.

\subsection{Phase extraction}

Appendix \ref{multiple} (see also \cite{a13}) indicates that the quantum state with  phases to be extracted  must be the following form  (that is,  Eq. (\ref{index_have_to}))
\begin{eqnarray} \label{apply_2}
|I(\{\phi_j\})\rangle = \frac{1}{\sqrt{d}} (|0\rangle +e^{i\phi_1}|1\rangle + ... +e^{i\phi_{d-1}}|d-1\rangle),
 \end{eqnarray}
which is not equivalent to the form of the state representation of quantum images as Eq. (\ref{eq_def}) intuitively.
In order to obtain the phases of the state of the quantum image, we should transform the representation of the quantum image from the form as Eq. (\ref{eq_def}) to a similar form as Eq. (\ref{apply_2}).

To this end, note that Eq. (\ref{eq_def}) could be expressed as the following form,
\begin{eqnarray} \label{deformation_1}
|I( \{ \theta_j \} )\rangle = \frac{1}{2^n} \sum_{j=0}^{2^{2n}-1} |0\rangle|j\rangle + \frac{1}{2^n} \sum_{j=0}^{2^{2n}-1} e^{i\theta_j}|1\rangle|j\rangle.
\end{eqnarray}
Let the states associated with phases in Eq. (\ref{deformation_1})  be a set $A$, then
\begin{align}
A = \{ |1\rangle|0\rangle, |1\rangle|1\rangle, |1\rangle|2\rangle,  ..., |1\rangle|2^{2n}-1\rangle \},
\end{align}
where the left qubit $|1\rangle$ in $A$ is a computational basis.
Given another set $B$,
\begin{align}
B=\{|1\rangle, |2\rangle, |3\rangle, ..., |2^{2n}\rangle\},
\end{align}
where $|1\rangle$ in $B$ is not the computational basis. Encoding $A$ to $B$, we get

\begin{eqnarray} \label{deformation_2}
|I( \{ \theta_j \} )'\rangle = \frac{1}{2^n} \sum_{j=0}^{2^{2n}-1} |0\rangle|j\rangle + \frac{1}{2^n} \sum_{j=1}^{2^{2n}-1} e^{i\theta_j}|j\rangle,
\end{eqnarray}
where,
as far as phase estimation  concerned, Eq. (\ref{deformation_2}) is consistent with the form as Eq. (\ref{apply_2}), so MPE can be applied on Eq. (\ref{deformation_2}) to extract its phases.

\subsection{Rotation operator constructed}
\label{overflow}

Common sense indicates that if phase estimation is reduced, then the error of measurement is reduced also.
Naturally, comparing with the phase error introduced by applying MPE on the two kinds of states of quantum images involved, if we do not estimate the phases of the carrier, namely using phases estimated from embedder only to construct a rotation operator can reduce the error in the image synthesis, but the shortcoming of this method is also obvious.
Actually, since we do not know the phase information of carrier, to control the result of the phase addition is impossible. Therefore the cost of this measure may increase the risk of phase overflow.
Appendix \ref{shortcoming} is the interpretation of such an example.


According to our basic requirements,  phase which reaches or exceeds $\frac{\pi}{2}$ is an exception.
In order to avoid exception,  an unitary transformation $U_n(\theta_j', \varphi_j'), j \in \{0, 1, ..., 2^{2n}-1\}$, is constructed and used to restrain the phases in $(0, \frac{\pi}{2})$, where $\theta_j'$ and $\varphi_j'$ denote the estimated phases of the embedder and the carrier, respectively.

In order to guarantee that the synthesized pixels are in $(0, \frac{\pi}{2})$  as much as possible,
a monotone increasing function $\tanh (x)$ is required,
\begin{eqnarray}
\tanh(x)=\frac{e^{2x} - 1}{e^{2x} + 1} \in (-1,1), \ \frac{\pi}{2} \tanh(x) \in (-\frac{\pi}{2},\frac{\pi}{2}).
\end{eqnarray}
$\tanh(x)$ is monotone increasing when $\tanh(x) \in [0, 1]$. Numerical simmulation shows that when $x \geq 3$, $\tanh(x) \approx 1$.


Taking a two dimensional transform $U_2$ as an example (namely, the image just has one pixel point).
Assume that the state of the embedder and carrier are $|0\rangle +  e^{i\theta}|1\rangle$ and $|0\rangle +  e^{i\varphi}|1\rangle$, respectively.
On the other hand, let the phase estimated from the state $|0\rangle +  e^{i\theta}|1\rangle$   be $\theta'$, and the phase estimated from the state $|0\rangle +  e^{i\varphi}|1\rangle$ be  $\varphi'$.
We define operator $U_2(\theta', \varphi')$ with the following transform,
\begin{eqnarray} \label{mechanism_a}
U_2(\theta', \varphi')(|0\rangle +  e^{i\varphi}|1\rangle) \rightarrow  |0\rangle + e^{i[\frac{\pi}{2} \tanh(g(\theta', \varphi'))+\varphi']}|1\rangle,
\end{eqnarray}
where we define
\begin{eqnarray} \label{further_a}
g(\theta', \varphi') \doteq  \theta' + \varphi',
\end{eqnarray}
then, the synthesized result would be in $(0, \frac{\pi}{2})$ with a high possibility when the error of $\theta'$ and $\varphi'$ as small as possible.
The phase accumulation with Eq. \ref{mechanism_a}  degrading the risk of overflow could be proved in section \ref{restriain}.

In Eq. (\ref{mechanism_a}), $U_2(\theta',  \varphi')$ is unitary, this is because
\begin{eqnarray} \label{subtract}
 U_2(\theta', \varphi') e^{i \varphi}|1\rangle \rightarrow  e^{i[\frac{\pi}{2} \tanh( g(\theta', \varphi')) - \varphi'+\varphi]}|1\rangle,
\end{eqnarray}
and
\begin{align}
U_2(\theta', \varphi') |0\rangle = |0\rangle.
\end{align}
That is,
\begin{eqnarray}
 U_2(\theta', \varphi') |1\rangle \rightarrow  e^{i[\frac{\pi}{2} \tanh(g(\theta', \varphi') ) - \varphi']}|1\rangle,
\end{eqnarray}
we have
\begin{eqnarray}
 U_2(\theta', \varphi') |1\rangle \langle 1| \rightarrow  e^{i[\frac{\pi}{2} \tanh( g(\theta', \varphi') ) - \varphi']}|1\rangle \langle 1|.
\end{eqnarray}
On the other way, since $U(\theta', \varphi') |0\rangle  \rightarrow |0\rangle$,
\begin{eqnarray} \label{index_a}
U_2(\theta',  \varphi') = e^{i[\frac{\pi}{2} \tanh( g(\theta', \varphi') )  - \varphi']}|1\rangle \langle 1| + |0\rangle \langle 0|.
\end{eqnarray}
Thus we have
\begin{eqnarray}
U_2(\theta',  \varphi')U_2(\theta',  \varphi')^{\dagger} = I.
\end{eqnarray}
Therefore, $U_2(\theta',  \varphi')$ is unitary.

The approach above means that we should get the phases of the embedder and the carrier before constructing this operator.
Once the phases of two different images could be estimated with higher precision, the overflow control would be better ( the limit is $0$ error, thus  $\varphi_j-\varphi_j' =0$ in Eq. (\ref{subtract})).
Actually, owing to the Heisenberg limit, the error in phase estimation always exists, so $\varphi_j-\varphi_j' =0$ is impossible.
About the effectiveness of overflow control, see section \ref{restriain}.
For two dimensional case, the element of rotation matrix $U_2(\theta',  \varphi')$ should be
\begin{eqnarray} \label{eq31}
U_2(\theta',  \varphi') =
      \left(
      \begin{array}{cc}
        1              &   0             \\
        0              &   e^{i [\frac{\pi}{2} \tanh( \theta' + \varphi') - \varphi']}
      \end{array}
    \right),
\end{eqnarray}

Eq. (\ref{eq31}) shows that it exists an unitary operator which can be used to restrain the phase overflow. We, now extend the dimension of $U_2(\theta',  \varphi')$ from $2$-dimension to $2^{2n+1}$-dimension to get $U'$, which can be used to transform the state of the quantum image as Eq. (\ref{eq_def}) to get a feasible goal state.
%
Constructing the operator $U'$ should estimate and get all the phases of the embedder and carrier.
Suppose that the estimated phases of two kinds of images ( that is, the carrier and the embedder ), to be synthesized are $\{\varphi_1', \varphi_2', ..., \varphi_{2^{2n}}'\}$ and $\{ \theta_1', \theta_2' ..., \theta_{2^{2n}}'\}$, respectively.
Then, the operator which could be employed to transform $2^{2n}$ dimensional image as Eq. (\ref{eq_def}) is as the following form.

\begin{eqnarray}  \label{unit}
U'=
\left[
\begin{array}{cc cc c cc cc cc c}
1_1         &   0               &   0       &   0               &   \cdots    &      0      &   0              & 0                 &  0                &0          &   0       & 0  \\
0           &   1_2             &   0       &   0               &   \cdots    &      0      &   0              & 0                 &  0                &0          &   0       & 0  \\
\vdots      &   \vdots          &  \vdots   &   \vdots          &   \ddots    &  \vdots     &   \vdots          &  \vdots   &   \vdots           &  \vdots          &  \vdots   &   \vdots\\
0           &   0               &   0       &   0               &   \cdots    & 1_{2^{2n}}  &   0              &  0                &   0               &   0       &   0       & 0  \\
0           &   0               &   0       &   0               &   \cdots    &      0      &   exp_1 &  0                &   0               &   0       &   0       & 0  \\
0           &   0               &   0       &   0               &   \cdots    &      0      &   0              &  exp_2   &   0               &   0       &   0       & 0 \\
0           &   0               &   0       &   0               &   \cdots    &      0      &  0               &   0               &   exp_3               &   0       &   0       & 0     \\
0           &   0               &   0       &   0               &   \cdots    &      0      &  0               &   0               &   0               &   exp_4       &   0       & 0     \\
\vdots      &   \vdots          &  \vdots   &   \vdots          &   \vdots  &   \vdots      &  \vdots   &   \vdots          &  \cdots           &   \cdots  &  \ddots   &  \cdots\\
0           &   0               &   0       &   0               &   \cdots    &      0      &  0               &   0               &   0               &   0       &   0       & exp_{2^{2n}},
\end{array}
\right]
\end{eqnarray}
where
\begin{eqnarray} \label{condition}
exp_1 = e^{i[\frac{\pi}{2} \tanh (\theta_1'+\varphi_1') - \varphi_1']}    \\
exp_2 = e^{i[\frac{\pi}{2} \tanh (\theta_2'+\varphi_2') -\varphi_2']}    \\
...   \nonumber  \\
exp_{2^{2n}} =  e^{i[\frac{\pi}{2} \tanh (\theta_{2^{2n}}'+\varphi_{2^{2n}}')-\varphi_{2^{2n}}']}
\end{eqnarray}

\subsection{State of the synthesized image}
Applied $U'$ on the state as Eq. (\ref{eq_def}), the state of the synthesized quantum image should be
\begin{eqnarray} \label{syn}
|res\rangle = \frac{1}{2^n} \sum_{j=0}^{2^{2n}-1} (|0\rangle +  e^{i[\frac{\pi}{2} \tanh (\theta_j + \varphi_j) + \delta_j]}|1\rangle) \bigotimes |j\rangle.
\end{eqnarray}
where $\delta_j = \varphi_j - \varphi_j'$ is an error item.
$\varphi_j$ is the original true phase at position $j$.
 $\delta_j$ is the difference between the true phase and the estimated phase of the carrier image.

 $\mathbf{Open \ problem}$
One of the critical step in algorithm \ref{algorithm_1} is estimation of phases of two states of quantum images. So MPE is introduced, and the error of the estimated values could not be avoided. This means that some errors in the synthesized result is inevitable. That's to say, the synthesis of quantum image could be distorted to some extent. Theoretically, this problem could be resolved via increasing the number operator to restrain the error (see section \ref{restriain}), after all, the error could not be eliminated.


\section{Analysis of synthesis operation}
\label{analysis}
\subsection{Computation complexity of the preparation of quantum images}
Since the measurement will lead to state collapse, to complete the phase estimation for each kind of image with pixels $2^n \times 2^n$, $O(2^{2n})$  particles \cite{a13} which  take the same phase information are needed.
The state preparation of quantum images as Eq. (\ref{eq_cos}) has already been discussed in \cite{a1},
\begin{eqnarray} \label{eq_cos}
|\psi( \{ \beta_j \} )\rangle  = \frac{1}{2^n} \sum_{j=0}^{2^{2n}-1} (\cos \beta_j |0\rangle +  \sin \beta_j |1\rangle) \bigotimes |j\rangle,
\end{eqnarray}
where $\cos \beta_j |0\rangle +  \sin \beta_j |1\rangle$ ( where $\beta_j \in [0, \frac{\pi}{2}]$, $(j \in  [0,2^{2n}-1] )$) encodes the pixel information at $j^{th}$ position.
 The process of the state preparation shows the following theorem \ref{theorem_1}.

\begin{theorem} \label{theorem_1}
[ Yan, Iliyasu and Jiang, \cite{a1}]Given an angle vector $\theta=(\theta_0,\theta_1,...,\theta_{2^{2n}-1})$, there is an unitary transform $\mathcal{P}$ that can be implemented by a quantum circuit with  polynomial number of Hadamard gates to transform the input state $|0\rangle^{2n+1}$ to a FRQI state Eq. (\ref{eq_cos}).
\end{theorem}
 Note, the proof of the theorem \ref{theorem_1} confers Appendix \ref{theorem1_proof}, or see also \cite{a1}.
To describe the state of a quantum image, even the essence of representing the quantum image are the same, but the form of Eq. (\ref{eq_def}) \cite{a7} (representing pixels with angles) is different from Eq. (\ref{eq_cos}) (representing pixels with phases).

\begin{theorem} \label{theorem_2}
[Song and Niu, \cite{a7}] Given a quantum image $|I( \{ \theta_j \} )\rangle = \\
 \frac{1}{2^{n}}  \sum_{j=0}^{2^{2n}-1}(|0\rangle + e^{i\theta_j}|1\rangle) \bigotimes |j\rangle$,
there is a $2n+1$ qubits unitary transform $C$  that transforms a quantum image $|I(\theta)\rangle$ to the quantum image $|I( \{ \psi_j \} )\rangle = \frac{1}{2^{n}}  \sum_{j=0}^{2^{2n}-1}(|0\rangle + e^{i\psi_j}|1\rangle) \bigotimes |j\rangle$.
\end{theorem}
Note, the proof of theorem \ref{theorem_2} is similar to the proof of theorem \ref{theorem_1}, so we omit its proof. In theorem \ref{theorem_2} ( see also \cite{a7}), Song et al. claim that  polynomial qubits are needed when we decompose the relation between $|I(\theta)\rangle$ and $|I(\psi)\rangle$.
Comparing with the proof procedure between theorem \ref{theorem_1} and \ref{theorem_2}, we give a complete procedure of preparing the state of quantum image from the state $|0\rangle^{\bigotimes 2n+1}$ in the following theorem \ref{theorem_3}.

\begin{theorem} \label{theorem_3}
For any initial input state $|0\rangle ^{2n+1}$, there exists an unitary operator $P$ to transform the initial input state to the state $|I( \{ \theta_j \} )\rangle = \frac{1}{2^{n}}  \sum_{j=0}^{2^{2n}-1}(|0\rangle + e^{i\theta_j}|1\rangle) \bigotimes |j\rangle$, and only a polynomial number of Hadamard gates are needed to complete this transformation.
\end{theorem}

\begin{proof}
 This theorem can be proven with the same way as theorem \ref{theorem_1} and theorem \ref{theorem_2}.
 In order to obtain the goal state $I(\theta_j)\rangle$, we assume that the initial state of the system is $|0\rangle ^{\bigotimes 2n+1} = |0\rangle \bigotimes |0\rangle ^{\bigotimes 2n}$.
 This task can be realized in the following two steps.

\emph{Step \ 1.} Applying Hadamard gate to the initial state, we then get
\begin{eqnarray}  \label{evolve_39}
|\psi\rangle &= (H \bigotimes H^{\bigotimes 2n}) (|0\rangle \bigotimes |0\rangle ^{\bigotimes 2n})  \nonumber \\
&=\frac{1}{2^{n+1}} (|0\rangle + |1\rangle) \bigotimes  \sum_{j=0}^{2^{2n}-1} |j\rangle.
\end{eqnarray}

\emph{Step \ 2.} Constructing and applying a rotation operator $R_z(\theta_k)$,
\begin{eqnarray}
   R_z(\theta_k) &=
      \left(
      \begin{array}{cc}
        1              &   0             \\
        0              &   e^{i \theta_k}
      \end{array}
    \right),
\end{eqnarray}
to rotate the phase at the $k^{th}$  pixel location of the quantum image, we get the following operator $R_k$.

\begin{eqnarray}
R_k=(I \bigotimes  \sum_{j=0,j \neq k}^{2^{2n}-1}  |j\rangle \langle j|) + R_z( \theta_k) \bigotimes |k\rangle \langle k|.
\end{eqnarray}

Notice that $R_k R_k^{\dagger} = I$ and therefore $R_k$ is an unitary operator.
Applying $R_k$ on $|\psi\rangle$, we get
\begin{eqnarray} \label{ind_1}
&R_k    \frac{1}{2^{(n+1)}}   [    (|0\rangle + |1\rangle ) \bigotimes  \sum_{j=0}^{2^{2n}-1} |j\rangle   ]    \\
&= \frac{1}{2^{(n+1)}}  [ (I \bigotimes  \sum_{j=0,j \neq k}^{2^{2n}-1}  |j\rangle)  + ( |0\rangle + e^{i\theta_j} |1\rangle) \bigotimes |j\rangle  ].
\end{eqnarray}

Assume that we have an operator $R_p$, which is similar to $R_k$. $R_p$ is applied to the result of last step, we have
\begin{eqnarray} \label{ind_2}
&R_p R_k   \frac{1}{2^{(n+1)}}  [    (|0\rangle + |1\rangle) \bigotimes  \sum_{j=0}^{2^{2n}-1} |j\rangle   ]  \\
&= \frac{1}{2^{(n+1)}}  [ (I \bigotimes  \sum_{j=0,j \neq k,p}^{2^{2n}-1}  |j\rangle)  + (|0\rangle + e^{i\theta_j}|1\rangle) \bigotimes |j\rangle    \nonumber  \\
& + ( |0\rangle + e^{i\theta_p} |1\rangle) \bigotimes |p\rangle.
\end{eqnarray}

It is obvious that we can design our goal state $|I( \{ \theta_j \} )\rangle$  using the above operators repeatedly. That's to say, the operator $P$ could be constructed as,
\begin{eqnarray}
P =   (\prod_{i=1}^{2^{2n}}  R_i).
\end{eqnarray}
 Since all $R_i$ are unitary, $ \prod_{i=1}^{2^{2n}} R_i$ is also unitary , and  $P = \prod_{i=1}^{2^{2n}} R_i $ is unitary.
 By induction of  Eq. (\ref{ind_1}) and Eq. (\ref{ind_2}), we have
 \begin{align}
|I( \{ \theta_j \} )\rangle = P|\psi\rangle.
 \end{align}
 The number of Hadamard gates used in the process of preparing  the state of the quantum image as Eq. (2) is $O(n)$. In summary, the claim holds.
\end{proof}

\subsection{Uncertainty relation of the synthesized pixel}
\label{joint_uncertainty}
Since covariance measurement is used to estimate the multiple phases of the state of quantum image (see also, \cite{a13}), Heisenberg limit is the reason of inevitable phase error. In this part, we explore the uncertainty relations of a single pixel of the synthesized image.

Actually, Holevo's theoretical analysis \cite{a9} shows that covariant measurement has an uncertainty relation between the number operator and phases.
Consider the complex random variable $e^{i\varphi}$ taking values on the unit circle $(-\pi, \pi)$. The variance then is
\begin{eqnarray}
 D\{ e^{i\varphi} \} = \int_{-\pi}^{\pi} | e^{i\varphi} - E\{e^{i\varphi}\}|^2 P(d\varphi)
\end{eqnarray}
 where $E\{ e^{i\varphi}\} = \int e^{i\varphi} P(d\varphi)$. Then the value of the uncertainty of $\varphi$ is formulated as
\begin{eqnarray} \label{variance}
 \Delta \{ \varphi\}^2 = \frac{D \{ e^{i\varphi} \}}{ |E\{ e^{ i\varphi }\}|^2 }.
\end{eqnarray}

 Let $\mathcal{H}$ be an infinite dimensional Hilbert space and $\{ |n\rangle;n=0,1,..., \}$ is a basis, and let $N$ be the number operator,
\begin{eqnarray} \label{def_N}
 N = \sum_{n=0}^{\infty} n|n\rangle\langle n|.
\end{eqnarray}
\begin{eqnarray} \label{def_Delta}
\Delta N = || (N- \bar{N})|\varphi\rangle||^2, \bar{N} = \langle \varphi| N |\varphi\rangle
\end{eqnarray}
(Generally, constant is also an operator, if we multiply it with an identity operator).
Then we have the following lemma.

\begin{lemma} \label{lemma_0}
[Holevo, \cite{a9}]For any covariant measurement $M$
\begin{eqnarray} \label{uncertain_1}
\Delta_M\{\varphi\}^2 \geq (1- \frac{1}{2}|\langle \varphi | 0\rangle|^2)^{-1} (\frac{1}{4(\Delta N)^2} + \frac{1}{2}|\langle \varphi|0\rangle|^2),
\end{eqnarray}
the following uncertainty relation holds
\begin{eqnarray} \label{uncertain_2}
\Delta_M\{\varphi\} \cdot \Delta N \geq \frac{1}{2}.
\end{eqnarray}
\end{lemma}
where $\Delta_M\{\varphi\}$ is same as the variance of Eq. (\ref{variance}), and $|\varphi\rangle = \sum_n \varphi_n|n\rangle$.

The detailed procedure of proving lemma \ref{lemma_0} confers Appendix \ref{lemma_proof} (see also \cite{a9}).
The uncertainty relation inequality (\ref{uncertain_2}) obviously holds, because $(1- \frac{1}{2}|\langle \varphi | 0\rangle|^2)^{-1} >1$, meanwhile $\frac{1}{2}|\langle \varphi|0\rangle|^2 > 0$.
 Inequality (\ref{uncertain_2}) is a general inequality relation about two objects $\Delta_M\{\varphi\}$ and $\Delta N$.

Since the phase obtained from the state of the quantum image complies with the uncertainty relation as inequality (\ref{uncertain_2}), the pixel of the synthesized image has some implied uncertainty relations.
Note that, the embedder ( whose true phase is represented with $\theta_j$) is embedded into the carrier( whose true phase is represented with $\varphi_j$).

%
So for the pixel of the synthesised image,
we have a  general uncertainty relation described as theorem \ref{theorem_4}.

\begin{theorem}  \label{theorem_4}
For synthesis operation of quantum images, given two  quantum images $|I( \{ \varphi_{j} \} )\rangle$ (carrier) and $|I( \{ \theta_{j} \} )\rangle$ (embedder), which are represented as the form as Eq. (\ref{eq_def}). $\varphi_j$ and $\theta_j$ are the phases of the state of two images $|I( \{ \varphi_{j} \} )\rangle$ and $|I( \{ \theta_{j} \} )\rangle$ at $j^{th}$ position, respectively.
Suppose that the number operator is $N=\sum_{n=0}^{\infty} n |n\rangle\langle n|$, and $\{|n\rangle, n=0, 1, ...\}$ is a basis of Hilbert space.
Measured $|I( \{ \varphi_{j} \} )\rangle$ and $|I( \{ \theta_{j} \} )\rangle$ to get the estimated phases $\varphi_j'$ and $\theta_j'$ corresponding to their true phase $\varphi_j$ and $\theta_j$ individually and respectively, then the lower bound of the joint uncertainty of the pixel of the synthesized image is $ \tanh(1) + \frac{1}{2}$.
\end{theorem}

\begin{proof}
To prove this result, it only needs to check that the following three conditions are true simultaneously:

($\uppercase\expandafter{\romannumeral 1}$) In this paper, the method used to estimate all the phases is the covariance measurement which is  provided by Macchiavello \cite{a13}( detailed information sees section Appendix \ref{multiple}). This is consistent with the method used in \cite{a9}.

($\uppercase\expandafter{\romannumeral 2}$)  The quantum state which is used for extracting the phase,  is consistent with the state definition of quantum images. That is, the state representation of quantum image
\begin{eqnarray}
|I( \{ \theta_j \} )\rangle = \frac{1}{2^{n}}  \sum_{j=0}^{2^{2n}-1} (|0\rangle + e^{i\theta_j}|1\rangle) \bigotimes |j\rangle
\end{eqnarray}
can be turned to the form as Eq. (\ref{deformation_1})£¬
\begin{eqnarray}
|I( \{ \theta_j \} )\rangle = \frac{1}{\sqrt{2^{2n}}} (|0\rangle +e^{i\theta_1}|1\rangle + ... ++e^{i\theta_{2^{2n}}}|2^{2n}\rangle)
\end{eqnarray}
The two points ($\uppercase\expandafter{\romannumeral 1}$) and ($\uppercase\expandafter{\romannumeral 2}$) guarantee that the phases could be extracted from the quantum state.

Since the $j^{th}$ position of carrier image and the $j^{th}$ position of embedder image are two independent physical system,
for a synthesized pixel point, to distinguish two different independent physical systems, we label the number operators used in two independent physical systems for the phase extraction as $N_1$ and $N_2$.
by applying lemma \ref{lemma_0}, the uncertainty relation for each physical system could be represented as,
\begin{eqnarray} \label{un_1}
\Delta_M\{\varphi\} \cdot \Delta N_1 \geq \frac{1}{2}, \  \  \Delta_M\{\theta\} \cdot \Delta N_2 \geq \frac{1}{2},
\end{eqnarray}
where $\Delta N_1$ and $\Delta N_2$ represent the variance of the number operator in the first and second physical system. Similarly,
$\Delta_M\{\varphi\}$ and $\Delta_M\{\theta\}$ are the variance of phases calculated with phases estimated from two kinds of states of the quantum images.

($\uppercase\expandafter{\romannumeral 3}$)
Let $var_j$ be the general phase formula of the synthesized image. It is known from Eq. (\ref{uncertain_2}) that the pixel of the synthesized image could be represented as
\begin{align} \label{un_3}
var_j = \frac{\pi}{2} \tanh (\theta_j' + \varphi_j') + \delta_j.
\end{align}
Note that, the first and most important is that, the pixel information of the synthesized image is taken by an independent physical system.
Secondly, there is a linear operation about the estimated phase $\varphi_j'$ in $\delta_j$.
Thirdly, there is a linear operation about two estimated value $\varphi_j'$ and $\theta_j'$. Namely, $\varphi_j' + \theta_j$.
Fourthly, there is a non-linear operation which is applied on $\theta_j'+\varphi_j'$. Namely, $\tanh(\theta_j'+\varphi_j')$.
Especially, to obtain the phases of the state of two different images is independent in the two different physical systems.
Thus we have
\begin{eqnarray} \label{un_4}
\Delta_M\{\varphi\} \cdot \Delta N_1  + \Delta_M\{\theta\} \cdot \Delta N_2 \geq 1.
\end{eqnarray}

Combined with the inequality (\ref{un_1}) and inequality (\ref{un_4}), the pixel uncertainty of the synthesized image has the relation
\begin{align} \label{916_taifeng}
\tanh (\Delta_M\{\varphi\} \cdot \Delta N_1  + \Delta_M\{\theta\} \cdot \Delta N_2) + \Delta_M\{\varphi\} \cdot \Delta N_1  \geq \tanh(1) + \frac{1}{2}.
\end{align}
So the result holds.
\end{proof}

The uncertainty relation reflected by inequality (\ref{uncertain_2})
could be summarized by the following figures.
Fig. \ref{unc} (a), the boundary conditions of satisfying the minimum uncertainty relation is when the quantum state of the system is coherent state.
Fig. \ref{unc} (b) and Fig. \ref{unc} (c) describe two uncertainty relations between $P$ (phase) and $N$ (number operator) where the components of the squeezed state are compressed.
We use Fig. \ref{unc}(a), Fig. \ref{unc}(b), Fig. \ref{unc}(c) to induct the meaning of Fig. \ref{unc} (d), where if the number operator $N$ is sufficient large (or enough large), the phase will be exactly estimated.

Before interpreting the meaning of the uncertainty relation in inequality (\ref{916_taifeng}), we emphasis that the inequality (\ref{916_taifeng}) describes an uncertainty relation of a pixel of a new independent physical system.
This is because the final phase of the state of the synthesized image ( see Eq. (\ref{un_3})) is generated by applying the unitary operator $U'$ (see  Eq. (\ref{unit})) on the state of quantum image (note, this quantum image state is a particle with phases information, see Eq. (\ref{eq_def})).
Based on this view, we simultaneously know that,  Eq. (\ref{916_taifeng}) reflects a joint uncertainty of a new physical system.
$\Delta_M\{\varphi\} \cdot \Delta N_1 \geq \frac{1}{2}$ and $\Delta_M\{\theta\} \cdot \Delta N_2 \geq \frac{1}{2}$ are two independent physical systems, and the uncertainty relations could be described with table \ref{single_one} and table \ref{single_two} (for number operator, $\uparrow$ means that the number operator becomes larger, for precision of phase, the sign $\downarrow$ means that the precision of the phase estimation decreases, and so on), respectively.
Finally, $\Delta_M\{\varphi\} \cdot \Delta N_1 \geq \frac{1}{2}$ and $\Delta_M\{\theta\} \cdot \Delta N_2 \geq \frac{1}{2}$ induce the joint uncertainty relation between $\theta_j'$, $\varphi_j'$ and $N_1$, $N_2$ with inequality (\ref{916_taifeng}) (see table \ref{synthesized_pixel}).

\begin{table}
\renewcommand{\arraystretch}{1.3}
\caption{\label{single_one} Uncertainty relation in single physical system for extracting a phase from carrier image  ($\Delta_M\{\varphi\} \cdot \Delta N_1 \geq \frac{1}{2}$)}
\begin{center}
\begin{tabu}{|c|c|}
\hline
 the trend of N            &    phase precision              \\
\hline
$N_1 \ \uparrow $          &    $\varphi_j$ $\uparrow$       \\
\hline
$N_1$ $\downarrow$         &    $\varphi_j$ $\downarrow$      \\
\hline
\end{tabu}
\end{center}
\end{table}

\begin{table}
\renewcommand{\arraystretch}{1.3}
\caption{\label{single_two} Uncertainty relation in single physical system for extracting a phase from embedder image ($\Delta_M\{\theta\} \cdot \Delta N_2 \geq \frac{1}{2}$)}
\begin{center}
\begin{tabu}{|c|c|}
\hline
 the trend of N            &    phase precision              \\
\hline
$N_2 \ \uparrow $          &    $\theta_j$ $\uparrow$       \\
\hline
$N_2$ $\downarrow$         &    $\theta_j$ $\downarrow$      \\
\hline
\end{tabu}
\end{center}
\end{table}

\begin{table}
\renewcommand{\arraystretch}{1.3}
\caption{\label{synthesized_pixel} The joint uncertainty relation of the synthesized pixel of the synthesized image. The joint uncertainty relation  between $N_1$, $N_2$ and $\theta_j' + \varphi_j'$ ($\tanh (\Delta_M\{\varphi\} \cdot \Delta N_1  + \Delta_M\{\theta\} \cdot \Delta N_2) + \Delta_M\{\varphi\} \cdot \Delta N_1  \geq \tanh(1) + \frac{1}{2}$)}
\begin{center}
\begin{tabu}{|c|c|c|}
\hline
 the trend of $N_1$     &  trend of $N_2$               &   precision of joint phase    \\
\hline
$N_1$  $\uparrow $       &   $N_2 $ $\uparrow$          &    $(\theta_j + \varphi_j)$ $\uparrow$       \\
\hline
$N_1$  $\downarrow$      &   $N_2 $ $\uparrow$          &    $(\theta_j + \varphi_j)$ uncertain     \\
\hline
$N_1$  $\uparrow $       &   $N_2 $ $\downarrow$        &    $(\theta_j + \varphi_j)$ uncertain       \\
\hline
$N_1$  $\downarrow$      &   $N_2 $ $\downarrow$        &    $(\theta_j + \varphi_j)$ $\downarrow$     \\
\hline
\end{tabu}
\end{center}
\end{table}

\begin{figure} 
\centering
\subfigure[]{\includegraphics[height=2cm,width=2cm]{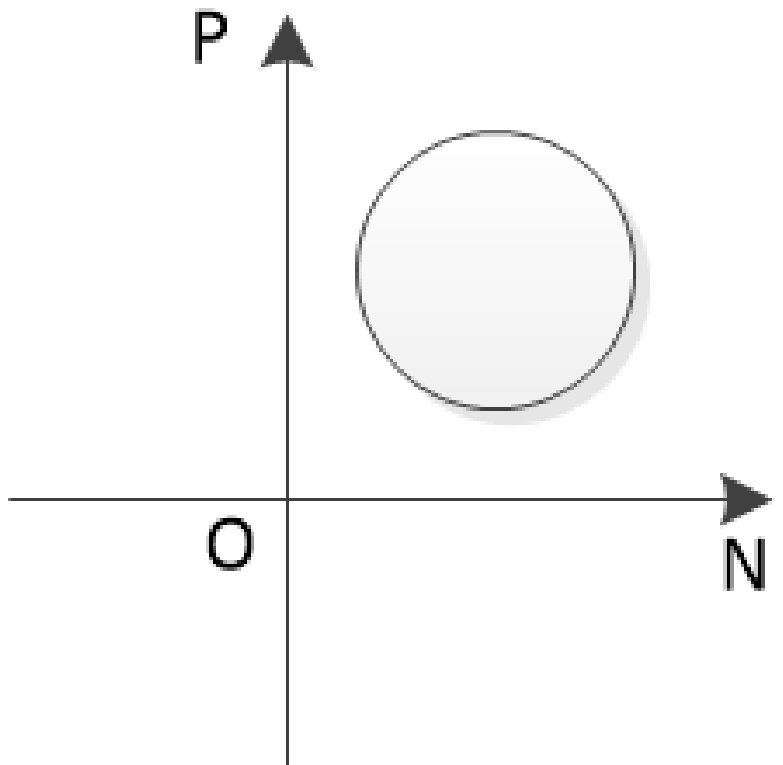}}
\subfigure[]{\includegraphics[height=2cm,width=2cm]{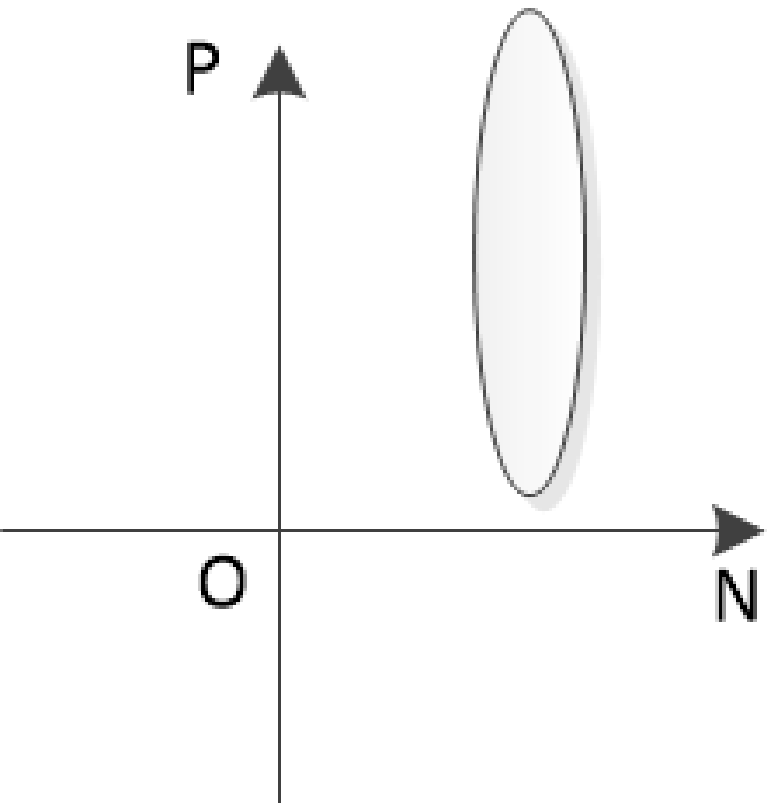}}
\subfigure[]{\includegraphics[height=2cm,width=2cm]{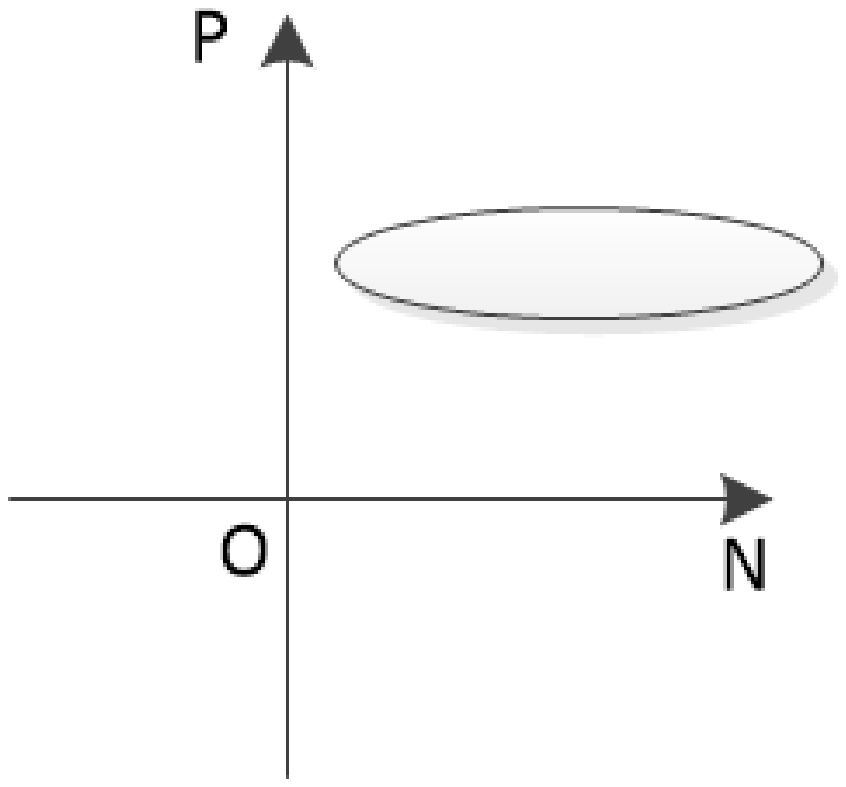}}
\subfigure[]{\includegraphics[height=2 cm,width=4cm]{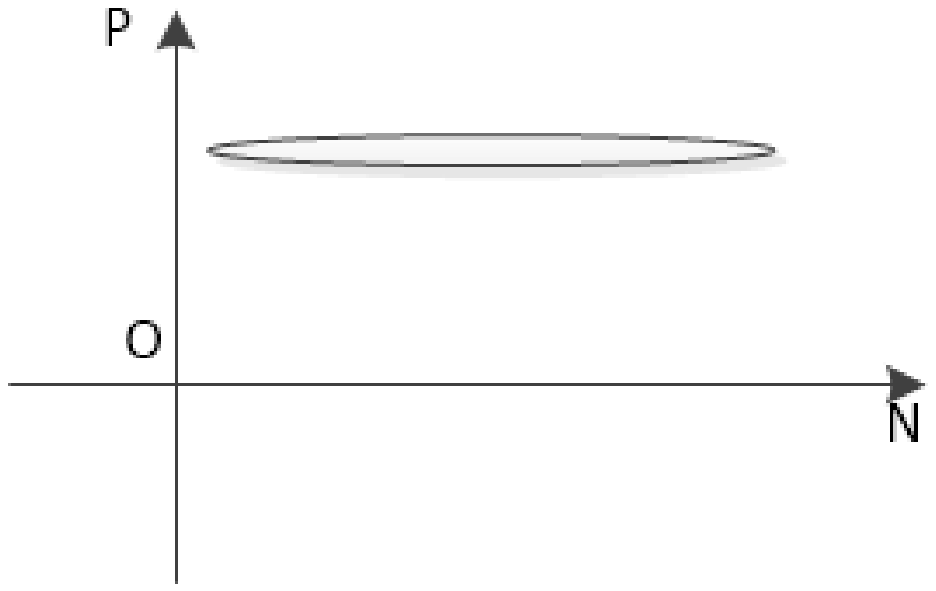}}
\caption{\label{unc} (a) Coherent state satisfies the equivalence of the lower bound. (b) and (c) describe two uncertainty relations where the orthogonal components of the squeezed state are compressed. (d) The uncertainty with larger particles number (represented by N) than (c).}
\end{figure}

\subsection{Effectiveness analysis of overflow control}
\label{restriain}

The core of algorithm \ref{algorithm_1} is to control the phase overflow, so the effectiveness of algorithm \ref{algorithm_1} is decided by the effectiveness of phase overflow control.
Once getting the phases which are estimated from some quantum states, we then can analyse these phases with classical methods, because these values are classical information.
On the other hand, note that the phases required in Eq. (\ref{eq_def}) are in $(0, \frac{\pi}{2})$, and the error of phase estimation always exists (see inequality (\ref{uncertain_2})), so the measurement results exceed $\frac{\pi}{2}$ is possible. For the sake of consistency and rationality, we should guarantee that all these phases exceeding $\frac{\pi}{2}$  remain in a legal range.

$\mathbf{Restictions \ 1}$
Assume that $\varphi_j'$ is the phase got by applying MPE on the state of quantum image  at the $j^{th}$ position. Thus we have the following restrictions:
\begin{equation}
\begin{cases}
\varphi_j' = \varphi_j' \ mod  \  \frac{\pi}{2},  \ if  \   \varphi_j'> \frac{\pi}{2} \\
\varphi_j' = \varphi_j',                          \ if  \   0 <\varphi_j' < \frac{\pi}{2}.
\end{cases}
\end{equation}
That is, $\varphi_j' \in (0, \frac{\pi}{2})$ is a mandatory requirement. let the estimated phases $\varphi_j'$ do operation $(\varphi_j' \mod \frac{\pi}{2})$, and the results also are labelled with $\varphi_j'$.

$\mathbf{Restictions \ 2}$
If the phases in the quantum state of the synthesized image exceed $\frac{\pi}{2}$, then these phases are called exception phases.

From Eq. (\ref{syn}), we know that the general phase expression in the state of the synthesized image is
\begin{align} \label{phase_syn}
\frac{\pi}{2} \tanh (\theta_j' + \varphi_j') + \delta_j.
\end{align}
According to the restrictions 2, if Eq. (\ref{phase_syn}) exceeds $\frac{\pi}{2}$, then the $j^{th}$ pixel is an exception pixel.

\begin{definition} \label{compression_def}
Let $\psi_{in}$ be the length of the phase space to be compressed, and $\psi_{out}$ be the length of the phase space compressed, then we call the ratio $\frac{\psi_{out}}{\psi_{in}}$ the compression ratio (of overflow control).
\end{definition}
From the common sense, if such a ratio $\frac{\psi_{out}}{\psi_{in}} <1$, then the strategy is effective. If $\frac{\psi_{out}}{\psi_{in}}=1$, then the strategy is ineffective.

\subsubsection{Compression ratio and effectiveness of the algorithm}

By compressing phase information to control the phase in the specified range $(0, \frac{\pi}{2})$,
the advantage of our control strategy is reflected by a compression ratio.
About the effectiveness of the overflow control, we have the following result.
\begin{theorem} \label{theorem_5}
Assume that $\theta_j'$ and $\varphi_j'$ are the $j^{th}$ estimated phase corresponding to their true phase $\theta_j$ and $\varphi_j$ of two different images, $(j \in \{1, 2, ..., 2^{2n}\})$, respectively.
Let $\delta_j =\varphi_j - \varphi_j'$, when the number operator $N_1$ and $N_2$ used in two physical systems are enough large, the limit of the compression ratio of phase is $\frac{1}{2}$.
\end{theorem}
\begin{proof}
According to the known conditions, such a ratio could be expressed
\begin{align} \label{ratio1}
\frac{\frac{\pi}{2} \tanh (\theta_j' + \varphi_j') + \delta_j}{\theta_j' + \varphi_j'}.
\end{align}
When $N_1$ and $N_2$ are enough large, then $\delta_j$ is an infinitesimal. The Eq. \ref{ratio1} approximates
\begin{align} \label{ratio2}
\frac{\frac{\pi}{2} \tanh (\theta_j' + \varphi_j')}{\theta_j' + \varphi_j'}.
\end{align}
Since
\begin{align}
\theta_j'+\varphi_j' \in (0, \pi),  \   \frac{\pi}{2} \tanh (\theta_j' + \varphi_j') \in  (0,  \frac{\pi}{2}),
\end{align}
that's to say, the interval $(0, \pi)$ is compressed into $(0, \frac{\pi}{2})$.
Thus the compression ratio is $\frac{\frac{\pi}{2}}{\pi} = \frac{1}{2}$.
So the result holds.
\end{proof}

Note that $\tanh(x)$ is a canonical function about $x$, and the feature of $\tanh(x)$ implies that,  the compression ratio in interval $(0,3)$ is greater than ratio in the interval $(3, \pi)$.
Therefore, the compression ratio is non-uniform when $x$ of $\tanh(x)$ changes in the interval $(0, \pi)$.

\begin{corollary} \label{corollary_1}
Assume that $N_1$ is a general number operator used for phase extraction of the carrier image, then the compression ratio is in the range $[0, \frac{1}{2})$.
\end{corollary}
\begin{proof}
When $N_1$ is a general number, it includes two cases.
Firstly, when $N_1$ is sufficient large, $\delta_j \approx 0$ with high probability. From theorem \ref{theorem_5}, we can conclude that the upper bound of the compression ratio approximates $\frac{1}{2}$.
Secondly, it is known from the inequality $\Delta_M\{\varphi_j'\} \cdot \Delta N_1 \geq \frac{1}{2}$, when $N_1$ is not sufficient large, we know $\delta_j \neq 0$ with high probability.
Since
\begin{align}
\varphi_j \in (0, \frac{\pi}{2}),   \   \varphi_j' \in (0, \frac{\pi}{2}),
\end{align}
(note, $\varphi_j$ is a true phase of the image at position $j$) then
\begin{align}
\delta_j = \varphi_j - \varphi_j' \in (-\frac{\pi}{2}, \frac{\pi}{2}).
\end{align}
Further, since $\theta_j'>0$ and $\varphi_j'>0$, we have
\begin{align} \label{two_bound}
\pi > \frac{\pi}{2} \tanh (\theta_j' + \varphi_j') + \delta_j > -\frac{\pi}{2}.
\end{align}
We do not consider the exception case of overflow. Then according to the definition of the compression ratio, inequality (\ref{two_bound}) includes the lower bound of the synthesized pixel (or phase information).
That is, if
\begin{align}
\theta_j'+\varphi_j' = \frac{\pi}{2} \tanh (\theta_j' + \varphi_j') + \delta_j,
\end{align}
namely, when
\begin{align} \label{ratio_lower}
\delta_j = \theta_j'+\varphi_j' - \frac{\pi}{2} \tanh (\theta_j' + \varphi_j'),
\end{align}
then, according to the definition \ref{compression_def}, $\psi_{out} = \theta_j'+\varphi_j' - ( \frac{\pi}{2} \tanh (\theta_j' + \varphi_j') + \delta_j) =0$.
So, $0$ compression ratio happens. Because the phases estimated satisfy the uncertainty relations, so such case is possible.
Therefore, the compression ratio is in the range $[0, \frac{1}{2})$. In a word, the result holds.

\end{proof}

\begin{corollary}  \label{corollary_2}
Achieving the maximal compression ratio is a sufficient but not a necessary conditions for $\frac{\pi}{2} \tanh (\theta_j' + \varphi_j') + \delta_j < \frac{\pi}{2}$.
\end{corollary}
\begin{proof}
Assume that the number operator used to estimate the $j^{th}$ phases of two physical systems are $N_1$ and $N_2$.
On the one hand,
according to the theorem \ref{theorem_5} and corollary \ref{corollary_1}, the maximal compression ratio means that, $N_1$ and $N_2$ are sufficient large, namely, $\delta_j \approx 0$.
Then, the phase of the state of the synthesized image at $j^{th}$ position approximates $\frac{\pi}{2} \tanh (\theta_j' + \varphi_j')$, ($j$ could be any one of $0, 1, 2, ..., 2^{2n}-1$).
Obviously, no matter what $\theta_j' + \varphi_j' \in (0, \pi)$ is, the synthesized result of the phase could not overflow.
On the other hand, whether the result of the synthesized phase overflow or not is determined by $\frac{\pi}{2} \tanh (\theta_j' + \varphi_j')$ and $\delta_j$ simultaneously.
It is apparent that the condition which satisfies  $\frac{\pi}{2} \tanh (\theta_j' + \varphi_j') + \delta_j < \frac{\pi}{2}$ if $\theta_j' + \varphi_j'$ and $\delta_j$ are not too large, such as, when $\theta_j' + \varphi_j' = \frac{\pi}{10}$, and $\delta_j=\frac{\pi}{10}$.
That's to say, $\frac{\pi}{2} \tanh (\theta_j' + \varphi_j') + \delta_j < \frac{\pi}{2}$ does not mean that the compression ratio achieves the maximal compression ratio.
So, the result holds.
\end{proof}

\subsubsection{Overflow risk and restraint mechanism}

In this section, $\theta_j'$ and $\varphi_j'$ are assumed to be the $j^{th}$ estimated phase corresponding to the true phase $\theta_j$ and $\varphi_j$ of two different images, $(j \in \{1, 2, ..., 2^{2n}\})$, respectively.
Intuitively, since
\begin{align}
\delta_j =\varphi_j - \varphi_j',   \    \theta_j, \varphi_j \in (0, \frac{\pi}{2}), \   \theta_j', \varphi_j' \in \{0, \frac{\pi}{2}\},
\end{align}
the mathematic relation  inequality (\ref{two_bound}) holds.
Especially, we note that $(-\frac{\pi}{2}, 0)$ and $[\frac{\pi}{2}, \pi)$ are the two phase overflow intervals for the state of the synthesized quantum image. Indeed, it seems to be unfortunate. However, through analyzing some numerical relations, we will known that, though these overflow cases are possible to happen, the actual fact is not so bad.

Actually, the risks
are degraded when the two parts $\frac{\pi}{2}\tanh(\theta_j'+\varphi_j')$ and $\delta_j$ are combined.
On the one side, owing to $\tanh(3) \approx 1$, thus when
\begin{align}
x \in (0, 3),  \   x+\Delta x   \    \in (0 ,3), \ and  \  x - \Delta x   \   \in (0 ,3),
\end{align}
the steep function curve of $\tanh(x)$ means that a little deviation $\Delta x$ will lead to a larger  difference.
That is,
\begin{align}
\frac{\pi}{2}\tanh(x + \Delta x) \gg \frac{\pi}{2}\tanh(x).
\end{align}
Similarly, we have
\begin{align}
\frac{\pi}{2}\tanh(x - \Delta x) \ll \frac{\pi}{2}\tanh(x).
\end{align}
On the other side, $\delta_j>0$ implies that $\varphi_j'$ is less than $\varphi_j$, so
\begin{align}
\frac{\pi}{2}\tanh(\theta_j'+\varphi_j) > \frac{\pi}{2}\tanh(\theta_j'+\varphi_j').
\end{align}
Namely, $\frac{\pi}{2}\tanh(\theta_j'+\varphi_j')$ decreases, and if $\theta_j'+\varphi_j' \in (0,3)$, $\frac{\pi}{2}\tanh(\theta_j'+\varphi_j')$ reduces more.
Correspondingly, $\delta_j<0$ means that $\varphi_j'$ is larger than $\varphi_j$, so
\begin{align}
\frac{\pi}{2}\tanh(\theta_j'+\varphi_j) < \frac{\pi}{2}\tanh(\theta_j'+\varphi_j').
\end{align}
Obviously, $\frac{\pi}{2}\tanh(\theta_j'+\varphi_j')$ increases, and if $\theta_j'+\varphi_j' \in (0,3)$, $\frac{\pi}{2}\tanh(\theta_j'+\varphi_j')$ increases more.
Thus the analysis could be summarized as follows.
If $\delta_j >0$, then $\frac{\pi}{2} \tanh (\theta_j' + \varphi_j')$ becomes smaller, $\delta_j$ becomes larger.
If $\delta_j <0$, then $\frac{\pi}{2} \tanh (\theta_j' + \varphi_j')$ becomes larger, $\delta_j$ becomes smaller.
So
there exists a compromise relation between $\frac{\pi}{2} \tanh (\theta_j' + \varphi_j')$ and $\delta_j$.
However, the degree of compromise could not be measured strictly with numerical relation, because all these relations remain the extension of the uncertainty relation (see inequality (\ref{uncertain_2})).
In brief, the phase overflow mechanism proposed in this paper degrades the possibility of happening phase overflow.
So to some extent,  algorithm \ref{algorithm_1} effectively restrains the possibility of the phase overflow.

\subsubsection{Effectiveness and uncertainty of the synthesized pixel}
Theorem \ref{theorem_5} shows that the maximal compression ratio exists and equals $\frac{1}{2}$.
Corollary \ref{corollary_1} points out that the range of the compression ratio is in the range $[0, \frac{1}{2})$.
Corollary \ref{corollary_2} discusses the properties of the maximal compression ratio. These three results indicates that the phase overflow control in algorithm \ref{algorithm_1} is effective on the whole.
Theorem \ref{theorem_5} also indirectly indicates that the operator constructed  as Eq. \ref{unit} to implement the phase rotation is rational.

The compression ratio is associated with the precision of the phase estimation.
Taken two corresponding pixels to be synthesized as an example,
Assume that the number operators used to extract the phases of two corresponding physical system are $N_1$ and $N_2$, if $N_1$ and $N_2$ becomes larger, then it indicates the higher precision of the phase estimation, the less $\delta_j$, and the larger compression ratio with higher probability.
So the compression ratio is determined by $N_1$ and $N_2$.
Owing to the uncertainty relation $\Delta_M \{ \varphi_j' \} \cdot \Delta\{N_1\} \geq \frac{1}{2}$, and according to the theorem \ref{theorem_5} and corollary \ref{corollary_1}, except for the exception case of phase overflow, the compression ratio of the overflow control in the synthesized image is in the range $[0, \frac{1}{2})$, but uncertain.
Because the phase obtained is larger or smaller than its true phase is a probability issue, not a deterministic issue.
So it is impossible to determine whether the specified pixel overflow or not.
However, for quantum image synthesis, if only the number operators used in the covariance measurement are sufficient large, we then could judge that the effectiveness of the synthesized pixels of the quantum image could be good.

 \section{Quantum versus classical image processing}
 \label{compare}
 In our quantum images synthesizing procedure, a quantum image with $2^{2n}$ phases as the input is taken, where $n$ is the number of pixels on vertical and horizontal coordinates, respectively. By applying MPE, $2^{2n}$ phases are obtained.
We compare the classical and quantum image processing from the following aspects.
 (1) Image edition is a general task in the classical image processing. The synthesis procedure of quantum image in this paper shows that quantum images can also be edited.
This paper illustrates what aspects should be considered if we want to implement the quantum image synthesis successfully.
 (2) Compared with classical image processing, quantum image processing relies on the quantum and classical methods simultaneously.
 (3) Since to process a large size image in a classical computer is very difficult, so modern image processing on the classical computers always depends on the deep learning network. Otherwise, it is hardly to deal with it effectively. However, this task could be implemented on a quantum computer with matrix computation.
 (4) The greatest advantage in quantum image processing is that much less of memory is needed for storing a quantum image.
  For a quantum image represented with a matrix of the size $2^n \times 2^n$, $2n + 1 =\log_2 ^ {(2^{2n})} + 1$ qubits is enough  to store phase (pixel) matrix.
  This means that the storage space needed for storing the pixel information of a quantum image is drastically reduced.
  (5) However, the auxiliary space remains $O(2^{2n})$ needed to ensure the phase precision when we estimate phases with MPE.

\section{Conclusions}
\label{conclusion}

This paper raises how to implement synthesis of quantum images, and a method is brought forward to resolve this problem.
Since the pixel is represented with phase,to obtain the phases of the quantum image to be synthesized is the first step of implementing the synthesis of quantum images.
Thus MPE is applied to obtain the multiple phases of the quantum images.
However, since the error could be introduced by MPE and phase addition operation could lead to the phase overflow,
a rotation operator which is embedded the overflow control mechanism  is constructed. Applying this rotation operator on the state of carrier, we could get the goal state of the synthesized quantum image.
Then, we compute the joint uncertainty relation of the pixel of the synthesized image. Based on this calculating result, the discussion about compression ratio shows that the overflow control mechanism proposed in this paper to reduce the possibility of overflow is effective.
In this paper, we try to give a quantum image processing method for synthesizing two quantum images, so some defects might be in it, such as the complexity.
Therefore, better algorithms are expected to improve the performance of the quantum image synthesis further in the future.

\section*{Acknowledgements}
The first author would thank Professor Guangping He for useful suggestions. This work was partly supported by the National Natural Science Foundation of China (Nos. 61572532, 61876195, 61272058), the Natural Science Foundation of Guangdong Province of China (No. 2017B030311011), and the Fundamental Research Funds for the Central Universities of China (Nos. 17lgjc24).

%

\appendix
\section{Supplementary materials}
\subsection{Multiple phase estimation}
\label{multiple}
Parallelism is an important attributes of quantum information processing and  due
to that we can also expect that quantum information processing  of phases can be done  simultaneously and efficiently. There are, naturally, quite  a few ways to deal with this problem. Humphreys et al. \cite{a14} proposed one method to implement the MPE via finding simultaneously estimates of $D$ phases of the state $|\psi_{\theta}\rangle = \sum_{k=1}^D \alpha_k e^{i \mathbf{N_k \cdot \theta}}|\mathbf{N_k}\rangle$, where $\mathbf{N_k}$ is a number operator and  $D$ is a configuration number. However, this approach can not be used here because our form of the state of the quantum image is not consistent with the state considered in \cite{a14}. On the other side, the method proposed in \cite{a13} can handle this problem. We will summarize the  main idea of Macchiavello's\cite{a13} in the following.

We consider the estimation theory of $M$ independent phases $\phi_j\ (j=1,...,M)$ through the unitary transformation
\begin{eqnarray}
\rho_{\{\phi_j\}} = exp (-i \sum_{j=1}^M \phi_j \hat{H_j} ) \rho_0 exp(i \sum_{j=1}^M \phi_j \hat{H_j} )
 \end{eqnarray}
where $\hat{H_j}$ represent M commuting self-adjoin operators which are defined on the Hilbert space $\mathcal{H}$ of the considered quantum system.
The vectors $\{|{n_j}\rangle\}$ denote now  eigenvectors corresponding to the eigenvalues $n_j$ of the  operator $\hat{H_j}$.

According to the general framework of quantum estimation theory, a cost function $\bar{C}$ of $C({\bar{\phi_j},{\phi_j}})$ is defined as
\begin{eqnarray}
\bar{C} = \int_{0}^{2\pi} d\phi_1  ... \int_0^{2\pi} d\phi_M p_0(\{ \phi_j\}) \int_0^{2\phi} d\bar{\phi_1} \ ... \  \int_0^{2\pi} C({\bar{\phi_j}}, {\phi_j}) p({\bar{\phi_j}}|{\phi_j}),
 \end{eqnarray}
and depends on  the set of the $M$ estimated values ${\bar{\{ \phi_j\}}}$ corresponding to the $M$ actual values ${\{ \phi_j\}}$. The estimation problem is reduced to the problem of minimizing the average cost $\bar{C}$ by optimizing POVM $d\mu({\bar{\phi_j}})$.
 In the view of the basic laws of quantum mechanics, the relation $\int d\mu({\bar{\phi_j}}) = I$ has to be satisfied. Since there is no contribution to the average cost, $d\mu_{\bot}({\phi_j})$ is out of considerations.  For the detailed definition see \cite{a13}. $p_0({\phi_j})$ and $p({\bar{\phi_j}}|{\phi_j})$ are prior probability densities for real  values ${\phi_j}$ and the conditional probability of estimating the set of values ${\bar{\phi_j}}$ given to real values of ${\phi_j}$.

When a general class of cost functions  is considered, then by Holevo's  outcomes, the optimal POVM takes the form
\begin{eqnarray}
d\mu_{||}({\phi_j}) = \frac{d\phi_1}{2\pi} ... \frac{d\phi_M}{2\pi}  |e({\phi_j})\rangle \langle e({\phi_j})|,
 \end{eqnarray}
where $|e({\phi_j})\rangle$ is defined as
\begin{eqnarray}
|e({\phi_j})\rangle =  \sum_{n_j} exp(i\sum_j n_j \phi_j) |\{n_j\}\rangle.
 \end{eqnarray}
In such a case, the multiple phases of the state
\begin{eqnarray} \label{index_have_to}
|I({\phi_j})\rangle = \frac{1}{\sqrt{d}} (|0\rangle +e^{i\phi_1}|1\rangle + ... ++e^{i\phi_{d-1}}|d-1\rangle)
 \end{eqnarray}
can be estimated with the general POVM
\begin{eqnarray}
|e({\phi_j})\rangle = \sum_{n_j} exp(i  \sum_{j=1}^{d-1} n_j \phi_j) |n_0,n_1,...,n_{d-1}\rangle_s
 \end{eqnarray}
with the fidelity
\begin{eqnarray}
F({\phi_j}) = |\langle \psi_0|\psi({\phi_j})\rangle|^2 = \frac{1}{d^2}[d+2\sum_{j=1}^{d-1} cos\phi_j   \\
+ 2 \sum_{j>k} cos(\phi_j-\phi_k)],
 \end{eqnarray}
where
\begin{eqnarray}
|\psi_0\rangle = \frac{1}{d^N} \sum_{n_j} \sqrt{\frac{N!}{n_0!n_1! ... n_{d-1}!}} |n_0,n_1,..., n_{d-1}\rangle_s
 \end{eqnarray}
and where $N$ is the number of states $|e({\phi_j})\rangle $ that need to  be prepared.

\subsection{Shortcoming of constructing operator by using embedder's phases}
\label{shortcoming}
This section gives another method to construct a phase rotation transform.
The main goal is to compare with Eq. (\ref{unit}) and Eq. (\ref{matrix}), and then we can explore the advantage or disadvantage for the different transforms.
Let the phase extracted from Eq. (\ref{deformation_2}) be $(\theta_1', \theta_2', ..., \theta_{2^{2n}}')$.
Thus using $(\theta_1', \theta_2', ..., \theta_{2^{2n}}')$ to construct a rotation operator $U$, we have

\begin{eqnarray} \label{matrix}
 U =
\left[
\begin{array}{cc cc c cc cc cc c}
1_1         &   0               &   0       &   0               &   \cdots    &      0      &   0              & 0                 &  0                &0          &   0       & 0  \\
0           &   1_2             &   0       &   0               &   \cdots    &      0      &   0              & 0                 &  0                &0          &   0       & 0  \\
\vdots      &   \vdots          &  \vdots   &   \vdots          &   \ddots    &  \vdots     &   \vdots          &  \vdots   &   \vdots           &  \vdots          &  \vdots   &   \vdots\\
0           &   0               &   0       &   0               &   \cdots    & 1_{2^{2n}}  &   0              &  0                &   0               &   0       &   0       & 0  \\
0           &   0               &   0       &   0               &   \cdots    &      0      &   e^{i\theta_1'} &  0                &   0               &   0       &   0       & 0  \\
0           &   0               &   0       &   0               &   \cdots    &      0      &   0              &  e^{i\theta_2'}   &   0               &   0       &   0       & 0 \\
0           &   0               &   0       &   0               &   \cdots    &      0      &  0               &   0               &   e^{i\theta_3'}               &   0       &   0       & 0     \\
0           &   0               &   0       &   0               &   \cdots    &      0      &  0               &   0               &   0               &   e^{i\theta_4'}       &   0       & 0     \\
\vdots      &   \vdots          &  \vdots   &   \vdots          &   \vdots  &   \vdots      &  \vdots   &   \vdots          &  \cdots           &   \cdots  &  \ddots   &  \cdots\\
0           &   0               &   0       &   0               &   \cdots    &      0      &  0               &   0               &   0               &   0       &   0       & e^{i\theta_{2^{2n}}'},
\end{array}
\right]
\end{eqnarray}
where $1_1,1_2, ..., 1_{2^{2n}}$ are $1$. The index in each $1$ denotes how many $1$ used on the diagonal line (there are $2^{2n}$ $1$ in total ).
Apparently, $U U^{\dagger} = I$, so $U$ is unitary.

Applying $U$ on $|I(\varphi_j)\rangle$, we get the phase accumulation result $|res\rangle = U |I(\theta_j)\rangle$,
\begin{eqnarray}
|res\rangle = \frac{1}{2^n} \sum_{j=0}^{2^{2n}-1} |0\rangle|j\rangle + \frac{1}{2^n} \sum_{j=0}^{2^{2n}-1} e^{i(\theta_j' + \varphi_j)}|1\rangle|j\rangle.
\end{eqnarray}
That is,
\begin{eqnarray} \label{goal}
|res\rangle = \frac{1}{2^n} \sum_{j=0}^{2^{2n}-1} (|0\rangle +  e^{i(\theta_j' + \varphi_j)}|1\rangle) \bigotimes |j\rangle.
\end{eqnarray}
The analysis below will show that some problems exist if we directly apply $U$ on the state of quantum images.

According to the state representation of quantum image ( see Eq. (\ref{eq_def})), and combing with the restrictions 1 and restrictions 2 in section \ref{restriain}, for Eq. (\ref{goal}),
we have
\begin{align}
\varphi_j, \theta_j' \in (0, \frac{\pi}{2}), j=\{0, 1, ..., 2^{2n}-1\}.
\end{align}
However, according to the phase requirement about state representation of quantum images, we have
 \begin{align} \label{discuss_abnormal}
 \varphi_j + \theta_j' \in (0, \frac{\pi}{2}), j=\{0, 1, ..., 2^{2n}-1\}.
 \end{align}
That is, the allowable maximal phase is $\frac{\pi}{2}$(the upper bound of the phase accumulation is determined by Eq. (\ref{eq_def})).
If the phases in Eq. (\ref{eq_def}) which are not in the range $(0, \frac{\pi}{2})$ are undefined.
However, the addition in Eq. (\ref{discuss_abnormal}) is a common mathematic operation about two real values, and the result should be
\begin{align} \label{discuss_normal}
 \varphi_j + \theta_j' \in (0, \pi), j \in \{0, 1, ..., 2^{2n}-1\}.
\end{align}
The contradiction between Eq. (\ref{discuss_abnormal}) and Eq. (\ref{discuss_normal}) indicates that the overflow happens with high risks. If nothing measures are to be taken, it is probable that the synthesized phases exceed $\frac{\pi}{2}$.

Two points should be emphasised. Firstly, since the phase estimation for the case of $0$ phase (pixel is $0$), the outcome will be empty, the synthesis operation above will not affect the case of 0 phases.
Secondly, $\theta_j'+ \varphi_j$ in Eq. (\ref{goal}) may exceed $\frac{\pi}{2}$ if nothing measure is to be taken.

Compared with Eq. (\ref{unit}), the advantage of this method is reduce the times of measurement, so the error of phase is reduced one half. But the disadvantage is also very obvious, that is, there is nothing we can do to restrain the overflow when the overflow happens. For example, we do not introduce $\tanh(x)$ to Eq. (\ref{matrix}), because one half of phases are unknown. This is the reason why we choose Eq. (\ref{unit}) as the phase rotation transform.

 \subsection{Proof of Theorem \ref{theorem_1}}
\label{theorem1_proof}
 The following procedures of proving theorem 1 is digested from \cite{a1}. Two steps implement the preparation of the quantum state.

 Step1. applying transform $\mathcal{H} = I \bigotimes H^{\bigotimes (2n+1)}$, and the result is assumed as $|S\rangle$, we have
 \begin{eqnarray}
 |S\rangle = \mathcal{H} |0\rangle \bigotimes |0\rangle^{\bigotimes 2n} = \frac{1}{2^n} |0\rangle \bigotimes \sum_{j=0}^{2n} |j\rangle.
 \end{eqnarray}
 Then, construct the rotation matrices $R_y(2\theta_j)$ (along the Y-axis by the angle $2\theta_j$) and controlled rotation matrices $R_j$, $(j=0,1,..., 2^{2n}-1)$,
 \begin{eqnarray}
 R_y(2\theta_j)=
 \left(
      \begin{array}{cc}
        \cos \theta_j               &   -\sin \theta_j             \\
        \sin \theta_j               &   \cos \theta_j
      \end{array}
    \right),
 \end{eqnarray}

 \begin{eqnarray}
 R_j=(I \bigotimes  \sum_{i=0,i\neq 1}^{2^{2n}-1}  |i\rangle \langle i| )  + R_y(2\theta_j) \bigotimes |j\rangle \langle j|.
 \end{eqnarray}

 Since $R_j R_j^{\dagger} = I^{2n+1}$, $R_j$ is unitary. Applying $R_k$ and $R_l R_k$ on $|S\rangle$ gives rise to the following result:
 \begin{eqnarray}
 R_k(|S\rangle) = R_k (\frac{1}{2^n} |0\rangle \bigotimes  \sum_{j=0}^{2n-1}  |i\rangle)   \nonumber \\
 = \frac{1}{2^n} [ |0\rangle \bigotimes   \sum_{i=0,i \neq k}^{2^{2n}-1}  |i\rangle \langle i| + ( \cos \theta_k |0\rangle + \sin \theta_k |1\rangle) \bigotimes |k\rangle],
 \end{eqnarray}

 \begin{eqnarray}
 &R_l R_k |S\rangle =  \frac{1}{2^n} [ |0\rangle \bigotimes   \sum_{i=0,i \neq k}^{2^{2n}-1}  |i\rangle \langle i|   \nonumber \\
 &+   ( \cos \theta_l |0\rangle + \sin \theta_l |1\rangle) \bigotimes |l\rangle] + ( \cos \theta_k |0\rangle + \sin \theta_k |1\rangle) \bigotimes |k\rangle].
 \end{eqnarray}

Thus we can conclude that $R |S\rangle = ( \prod_{i=0}^{2^{2n}-1}  R_i) |S\rangle$, and this is the final state what we intend to prepare for. The scale of the resource overhead is described as theorem \ref{theorem_1}.

 \subsection{ Proof of Lemma \ref{lemma_0}}
\label{lemma_proof}
This theory evolves from the uncertainty relation between the rotation angle $\theta$ and the angular momentum $L$ of the particles.

In this part, we'll introduce covariance measurement and uncertainty relation referred. The complete introduction confers \cite{a9}.

Let $\mathbf{G}$ be a locally compact transitive group of transformations of a parametric set $\mathbf{\theta}$, and $\{V_g\}$ a continuous unitary
ray representation of $\mathbf{G}$ in a Hilbert space $\mathcal{H}$. Let $M(d\theta)$ be a $\mathbf{\theta}-measurement$, that is a generalised resolution
of identity in $\mathcal{H}$ on Borel subsets of $\mathbf{\theta}$. A measurement $M(d\theta)$ is covariant with respect to $\{V_g\}$ if
\begin{eqnarray}
V_g^{\star} M(B)V_g = M(g^{-1}B),
\end{eqnarray}
for any Borel $B \subset \mathbf{\theta}$.
The  covariant measurement has the general form as
\begin{eqnarray} \label{appendix_1}
M(d\theta) = e^{iL \theta} P_0 e^{-iL \theta} \frac{d\theta}{2\pi},
\end{eqnarray}
where $P_o$ is a positive operator.

The optimal covariant measurement is defined as
\begin{eqnarray} \label{appendix_2}
\langle| M_{\star}(d\theta)|m'\rangle = e^{i(m-m')\theta} \frac{\varphi_m \cdot \bar{\varphi}}{|\varphi_m| \cdot |\varphi_{m'}|} \cdot \frac{d\theta}{2\pi}
\end{eqnarray}

The proof the lemma \ref{lemma_0} starts from the uncertainty of angular momentum.

Let $M_{\star}$ be a covariant measurement with both Bayes and minimax for any measure of  deviation, then
$E_M{e^{i\theta}} = \sum_{-l+1}^1 \bar{\psi_{m-1}} p_{m-1,m} \psi_m$, where $M$ represents the expectation of $P_M$.

Introducing the operators $E_{\mp} = \int^{\pi}_{-\pi} e^{{\pm i\theta}} M^{\star}(d\theta)$, so that
\begin{eqnarray} \label{appendix_3}
E_{-} = E_{+}^{\star}, \ E_{-}E_{+} = I-|l\rangle\langle l|, \  E_{+}E_{-} = I-|-l\rangle\langle-l|.
\end{eqnarray}

Further introducing cosine operator $C = \frac{1}{2}(E_{+}) + E_{-}$ and sine operator $S=\frac{i}{2}(E_{+} - E_{-})$, and we have
\begin{eqnarray} \label{appendix_4}
C^2 + S^2 = I - \frac{1}{2} [|l\rangle \langle l| + |-l\rangle \langle -l|], \ [C,S] = \frac{i}{2} [|-l\rangle\langle -l|-l\rangle \langle l|].
\end{eqnarray}

\begin{eqnarray} \label{appendix_5}
E_{M_\star}\{e^{i\theta}\} = \langle \varphi| E_{-} | \varphi \rangle \equiv \bar{C^2} + \bar{S^2}, \\
|E_{M_{\star}} |^2  = \bar{C^2} + \bar{S^2}.
\end{eqnarray}
Using Eq. \ref{appendix_4}, we have
\begin{eqnarray} \label{appendix_6}
D_{M_{\star}}\{e^{i\theta}\} = ||(C - \bar{C})\varphi||^2 +  ||(S - \bar{S})\varphi||^2 + \frac{1}{2} (|\varphi_{-l}|^2 + |\varphi_l|^2)  \\
\equiv (\Delta C)^2 + (\Delta S)^2 + \frac{1}{2} (|\varphi_{-l}|^2 + |\varphi_l|^2).
\end{eqnarray}
Since $[C, L] = -iS$ and $[S, L] = iC$, the uncertainty relation satisfies with
\begin{eqnarray} \label{appendix_7}
(\Delta C)^2 + (\Delta S)^2 \geq \frac{1}{4} \bar{S^2},  \ (\Delta S)^2 (\Delta L)^2  \geq \frac{1}{4} \bar{C^2}.
\end{eqnarray}
By applying Eq. (\ref{appendix_4}), (\ref{appendix_5}), (\ref{appendix_6}), and inequality (\ref{appendix_7}), we obtain
\begin{eqnarray}
\Delta_{M_{\star}} \{\theta\}^2 \geq \frac{1}{4(\Delta)^2} + \frac{1}{2}(|\varphi_{-l}|^2 + |\varphi_l|^2) |E_{M_{\star}} \{ e^{i \theta}\}|^{-2}.
\end{eqnarray}
Based on Eq. \ref{appendix_2}, introducing phase operator
\begin{eqnarray}
P= \int_{-\pi}^{\pi} e^{i\varphi} M_{\star}(d \varphi), \  and  \   P^{\star} = \int_{-\pi}^{\pi} e^{-i\varphi} M_{\star}(d \varphi)
\end{eqnarray}
which has the relation
\begin{eqnarray}
PP^{\star} = I, \   P^{\star} P = I- |0\rangle\langle 0|.
\end{eqnarray}
With this condition, through analogizing the uncertainty relation about angular momentum, we can get
\begin{eqnarray}
\Delta_M\{\varphi\} \geq (1- \frac{1}{2}|\langle \phi | 0\rangle|^2)^{-1} (\frac{1}{4(\Delta N)^2} + \frac{1}{2}|\langle \phi|0\rangle|^2),
\end{eqnarray}
the following uncertainty relation holds
\begin{eqnarray}
\Delta_M\{\varphi\} \cdot \Delta N \geq \frac{1}{2}.
\end{eqnarray}

\end{document}